\documentclass[twoside]{article}%
\usepackage{amssymb}
\usepackage{amsfonts}
\usepackage{amsmath}
\usepackage{graphicx}%
\providecommand{\U}[1]{\protect\rule{.1in}{.1in}}
\newtheorem{theorem}{Theorem}

\newtheorem{corollary}[theorem]{Corollary}

\newtheorem{lemma}[theorem]{Lemma}

\newtheorem{remark}[theorem]{Remark}

\newenvironment{proof}[1][Proof]{\noindent\textbf{#1.} }{\ \rule{0.5em}{0.5em}}
\begin{document}

\title{\textbf{Blowup for }$C^{2}$\textbf{ Solutions of the }$N$\textbf{-dimensional
Euler-Poisson Equations in Newtonian Cosmology}}
\author{M\textsc{anwai Yuen\thanks{E-mail address: nevetsyuen@hotmail.com }}\\\textit{Department of Mathematics and Information Technology,}\\\textit{The Hong Kong Institute of Education,}\\\textit{10 Po Ling Road, Tai Po, New Territories, Hong Kong}}
\date{Revised 08-Nov-2013}
\maketitle

\begin{abstract}
Pressureless Euler-Poisson equations with attractive forces are standard
models in Newtonian cosmology. In this article, we further develop the
spectral dynamics method and apply a novel spectral-dynamics-integration
method to study the blowup conditions for $C^{2}$ solutions with a bounded
domain, $\left\Vert X(t)\right\Vert \leq X_{0}$, where $\left\Vert
\cdot\right\Vert $ denotes the volume and $X_{0}$ is a positive constant. In
particular, we show that if the cosmological constant $\Lambda<M/X_{0}$, with
the total mass $M$, then the non-trivial $C^{2}$ solutions in $R^{N}$ with the
initial condition $\Omega_{0ij}(x)=\frac{1}{2}\left[  \partial_{i}%
u^{j}(0,x)-\partial_{j}u^{i}(0,x)\right]  =0$ blow up at a finite time.

\ 

MSC: 35B30, 35B44, 35Q35, 85A05, 85A40

\ 

Key Words: Euler-Poisson Equations, Newtonian Cosmology, Initial Value
Problem, Blowup, Spectral-Dynamics-Integration Method, Attractive Forces,
$C^{2}$ Solutions, Bounded Domain, $R^{N}$

\end{abstract}

\section{Introduction}

The evolution of Newtonian cosmology can be modelled by the compressible
pressureless Euler-Poisson equations in dimensionless units:
\begin{equation}
\left\{
\begin{array}
[c]{rl}%
{\normalsize \rho}_{t}{\normalsize +\nabla\cdot(\rho u)} & {\normalsize =}%
{\normalsize 0}\\
\rho\lbrack u_{t}+(u\cdot\nabla)u] & {\normalsize =-}{\normalsize \rho
\nabla\Phi}\\
{\normalsize \Delta\Phi(t,x)} & {\normalsize =}{\normalsize \rho-\Lambda,}%
\end{array}
\right.  \label{Euler-Poissonnew}%
\end{equation}
where $\rho=\rho(t,x)\geq0$ and $u=u(t,x)\in\mathbf{R}^{N}$ are the density
and the velocity, respectively, with a background or cosmological constant
$\Lambda$.

The pressureless Euler-Poisson equations are the standard model in cosmology
\cite{FT}. If the Euler-Poisson equations include the pressure term, then they
provide the classical description of galaxies or gaseous stars in astrophysics
\cite{BT} and \cite{longair}. For details of the connection between the
Euler-Poisson equations (\ref{Euler-Poissonnew}) and Einstein's field
equations%
\begin{equation}
R_{\mu\nu}-\frac{1}{2}g_{\mu\nu}R-g_{\mu\nu}\Lambda=\frac{8\pi G}{c^{4}}%
T_{\mu\nu},
\end{equation}
where $R_{\mu\upsilon}$ is the Ricci curvature tensor, $R$ is the curvature
scalar, $g_{\mu\nu\text{ }}$ is the metric tensor, $T_{\mu\nu}$ is the
energy-momentum tensor of the universe, $G$ is Newton's gravitational constant
and $c$ is the light speed, interested readers can refer to Chapter 6 of
Longair's book \cite{longair}. In addition, for a geometrical explanation of
Newtonian cosmology with a cosmological constant ${\normalsize \Lambda}$,
interested readers can refer to Brauer, Rendall and Reula's paper
\cite{Brauer1994}.

For an analysis of stabilities for the related systems, interested readers can
refer to \cite{E}, \cite{M2}, \cite{MP}, \cite{BrauerJMP1998}, \cite{ELT},
\cite{DLY}, \cite{LT}, \cite{Y2}, \cite{CHAET}, \cite{CHENGT}, \cite{YuenNA}
and \cite{YuenBlowupEPCaAttractive}. In addition, there are explicit blowup or
global (periodical) solutions for the Euler-Poisson systems \cite{GW},
\cite{M1}, \cite{DXY}, \cite{Y1} and \cite{YUENCQG2009}.

It should be noted that in 2008, Chae and Tadmor \cite{CHAET} determined the
finite time blowup for the pressureless Euler-Poisson equations with
attractive forces (\ref{Euler-Poissonnew}) with $\Lambda=0$, under the initial
condition,%
\begin{equation}
S:=\{\left.  x_{0}\in R^{N}\right\vert \text{ }\rho_{0}(x_{0})>0,\text{
}\Omega_{0}(x_{0})=0,\text{ }\operatorname{div}u(0,x_{0})<0\}\neq\phi,
\end{equation}
where $u=(u^{1},u^{2},....,u^{N})$ and $\Omega_{0}(x_{0})$ is the re-scaled
vorticity matrix defined by $\Omega_{0ij}(x_{0})=\frac{1}{2}\left[
\partial_{i}u^{j}(0,x_{0})-\partial_{j}u^{i}(0,x_{0})\right]  $. Using
spectral dynamics analysis, they identified the Riccati differential
inequality%
\begin{equation}
\frac{D\operatorname{div}u(t,x_{0}(t))}{Dt}\leq-\frac{1}{N}\left[
\operatorname{div}u(t,x_{0}(t))\right]  ^{2}, \label{ineq1111new}%
\end{equation}
along the characteristic curve $\frac{dx_{0}(t)}{dt}=u(t,x_{0}(t))$. The
corresponding solution of the inequality (\ref{ineq1111new}) blows up at or
before $T=-N/\operatorname{div}u(0,x_{0}(0))$ with an initial condition that
requires that $\operatorname{div}u(0,x_{0}(0))$ at some non-vacuum state. An
improved blowup condition for the Euler-Poisson equations (\ref{ineq1111new})
was obtained by Cheng and Tadmor \cite{CHENGT} in 2009.

In this article, we modify the spectral dynamics method to introduce a
\textit{spectral-dynamics-integration} method for a bounded domain $X(t)$, to
obtain the new blowup conditions according to the following theorem.

\begin{theorem}
\label{thm:1}For the $N$-dimensional Euler-Poisson equations
(\ref{Euler-Poissonnew}), consider $C^{2}$ solutions with a bounded domain:
$\left\Vert X(t)\right\Vert \leq V_{\sup}$, where $\left\Vert \cdot\right\Vert
$ denotes the volume and $V_{\sup}$ is a positive constant. We define the
weighted functional%
\begin{equation}
H(t)=\int_{X(t)}\operatorname{div}ud\mu_{t}, \label{WeightedFucntion}%
\end{equation}
with the positive measure $d\mu_{t}=\rho(t,x(t))dx(t)$. If the initial
condition
\begin{equation}
\Omega_{0ij}(x)=\frac{1}{2}\left[  \partial_{i}u^{j}(0,x)-\partial_{j}u_{{}%
}^{i}(0,x)\right]  =0
\end{equation}
and any one of the following conditions\newline(1) $\Lambda<M/V_{\sup}%
$,\newline(2)\ $\Lambda\geq M/V_{\sup}$ and $H(0)<-\sqrt{-\frac{M^{3}%
N}{V_{\sup}}+\Lambda M^{2}N}$,\newline with the total mass $M=\int_{X(0)}%
\rho(0,x)dx>0$ are satisfied, the non-trivial $C^{2}$ solutions blow up at a
finite time $T.$
\end{theorem}

Here, the functional (\ref{WeightedFucntion}) represents the aggregate
density-weighted divergence of the velocity $u(t,x)$.

\section{New Spectral-Dynamics-Integration Method}

Before we present the novel spectral-dynamics-integration method, we first
quote the following lemma:

\begin{lemma}
[Proposition 2.2 on page 27 of \cite{TM}]\label{LemmaTM}Let $S$ be a material
system that fills the domain $X(t)$ at time $t$, and let $C$ be a function of
class $C^{1}$ in $t$ and $x.$ Then,%
\begin{equation}
\frac{d}{dt}\int_{X(t)}C(t,x)\rho(t,x)dx=\int_{X(t)}\frac{DC(t,x)}{Dt}%
\rho(t,x)dx,
\end{equation}
where $(D/Dt)=(\partial/\partial t)+u\cdot\bigtriangledown$ is the convective derivative.
\end{lemma}

In the following proof, we modify the method of spectral dynamics described in
\cite{LT}, \cite{CHAET} and \cite{CHENGT} to obtain the different blowup
conditions for the $C^{2}$ solutions.

\begin{proof}
[Proof of Theorem \ref{thm:1}]As the mass equation (\ref{Euler-Poissonnew}%
)$_{1}$:%
\begin{equation}
\frac{D\rho}{Dt}+\rho\nabla\cdot u=0,
\end{equation}
with the convective derivative,%
\begin{equation}
\frac{D}{Dt}=\frac{\partial}{\partial t}+\left(  u\cdot\nabla\right)
\end{equation}
could be integrated as:%
\begin{equation}
\rho(t,x_{0})=\rho_{0}(x_{0}(0,x_{0}))\exp\left(  -\int_{0}^{t}\nabla\cdot
u(t,x_{0}(t;x_{0}))dt\right)  \geq0 \label{positiveness}%
\end{equation}
for $\rho_{0}(x_{0}(0,x_{0}))\geq0$, the density function $\rho(t,x(t;x))$
generally conserves its non-negative nature.

For the momentum equations (\ref{Euler-Poissonnew})$_{2}$ and the solutions
with non-vacuum, we have%
\begin{equation}
u_{t}+u\nabla\cdot u=-\nabla\Phi.
\end{equation}
We take the divergence to the above equation to obtain:%
\begin{equation}
\nabla\cdot\left(  u_{t}+u\nabla\cdot u\right)  =-\Delta\Phi.
\end{equation}
If the initial condition $\Omega_{0ij}(x)=\frac{1}{2}\left[  \partial_{i}%
u^{j}(0,x)-\partial_{j}u^{i}(0,x)\right]  =0$ is satisfied, we can show by the
standard spectral dynamics in \cite{CHAET} and \cite{CHENGT} (by directly
applying equation (2.6) in \cite{CHAET} or equation (4.1) in \cite{CHENGT})
that
\begin{equation}
\frac{D}{Dt}\operatorname{div}u(t,x(t))+\frac{1}{N}\left[  \operatorname{div}%
u(t,x(t))\right]  ^{2}\leq-\rho(t,x(t))+\Lambda.
\end{equation}

We notice that the advancement in this article for the new blowup conditions
begins here. First, we multiply the density function $\rho(t,x(t))$ on both
sides and take the integration over the domain $X(t)$ to obtain:%
\begin{equation}
\rho(t,x(t))\left(  \frac{D}{Dt}\operatorname{div}u(t,x(t))+\frac{1}{N}\left[
\operatorname{div}u(t,x(t))\right]  ^{2}\right)  \leq-\left[  \rho
(t,x(t))\right]  ^{2}+\Lambda\rho(t,x(t)).
\end{equation}%
\begin{equation}
\int_{X(t)}\rho\left(  \frac{D}{Dt}\operatorname{div}u\right)  dx+\frac{1}%
{N}\int_{X(t)}\rho\left(  \operatorname{div}u\right)  ^{2}dx\leq-\int
_{X(t)}\rho^{2}dx+\Lambda\int_{X(t)}\rho dx
\end{equation}%
\begin{equation}
\int_{X(t)}\rho\left(  \frac{D}{Dt}\operatorname{div}u\right)  dx+\frac{1}%
{N}\int_{X(t)}\rho\left(  \operatorname{div}u\right)  ^{2}dx\leq-\int
_{X(t)}\rho^{2}dx+\Lambda M,
\end{equation}
where $M=\int_{X(t)}\rho dx=\int_{X(0)}\rho(0,x(0))dx(0)>0$ for non-trivial
solutions is the total mass of the fluid.\newline We apply Lemma \ref{LemmaTM}
with $C(t,x):=\operatorname{div}u$ to obtain
\begin{equation}
\frac{d}{dt}\int_{X(t)}\rho\operatorname{div}udx+\frac{1}{N}\int_{X(t)}%
\rho\left(  \operatorname{div}u\right)  ^{2}dx\leq-\int_{X(t)}\rho
^{2}dx+\Lambda M.
\end{equation}
We define the weighted functional%
\begin{equation}
H:=H(t)=\int_{X(t)}\operatorname{div}ud\mu_{t}%
\end{equation}
with the positive measure $d\mu_{t}=\rho(t,x(t))dx(t)$ for $\rho(0,x)\geq0$,
with the equation (\ref{positiveness}) to obtain%
\begin{equation}
\frac{d}{dt}H\leq-\frac{1}{N}\int_{X(t)}\left(  \operatorname{div}u\right)
^{2}d\mu_{t}-\int_{X(t)}\rho^{2}dx+\Lambda M. \label{Yuenineq}%
\end{equation}
We can estimate the first term on the right-hand side of the inequality
(\ref{Yuenineq}) to obtain
\begin{equation}
\left(  \int_{X(t)}\operatorname{div}ud\mu_{t}\right)  ^{2}=\left(  \left\vert
\int_{X(t)}\operatorname{div}ud\mu_{t}\right\vert \right)  ^{2}\leq\left(
\int_{X(t)}\left\vert \operatorname{div}u\right\vert d\mu_{t}\right)  ^{2}\leq
M\int_{X(t)}\left(  \operatorname{div}u\right)  ^{2}d\mu_{t}.
\end{equation}
Using the Cauchy-Schwarz inequality, we get%
\begin{equation}
\int_{X(t)}\left\vert \operatorname{div}u\right\vert d\mu_{t}\leq\left(
\int_{X(t)}1^{2}d\mu_{t}\right)  ^{1/2}\left(  \int_{X(t)}\left(
\operatorname{div}u\right)  ^{2}d\mu_{t}\right)  ^{1/2}%
\end{equation}%
\begin{equation}
\int_{X(t)}\left\vert \operatorname{div}u\right\vert d\mu_{t}\leq\left(
\int_{X(t)}\rho dx\right)  ^{1/2}\left(  \int_{X(t)}\left(  \operatorname{div}%
u\right)  ^{2}d\mu_{t}\right)  ^{1/2}=\sqrt{M}\left(  \int_{X(t)}\left(
\operatorname{div}u\right)  ^{2}d\mu_{t}\right)  ^{1/2}. \label{ineqlabc2}%
\end{equation}
We obtain%
\begin{equation}
\frac{\left(  \int_{X(t)}\operatorname{div}ud\mu_{t}\right)  ^{2}}{M}\leq
\int_{X(t)}\left(  \operatorname{div}u\right)  ^{2}d\mu_{t}%
\end{equation}%
\begin{equation}
-\frac{1}{N}\int_{X(t)}\left(  \operatorname{div}u\right)  ^{2}d\mu_{t}%
\leq\frac{-1}{MN}\left(  \int_{X(t)}\operatorname{div}ud\mu_{t}\right)
^{2}=\frac{-H^{2}}{MN}.
\end{equation}
The second term on the right-hand side of the inequality (\ref{Yuenineq}) can
be determined by%
\begin{equation}
\int_{X(t)}\rho dx\leq\left(  \int_{X(t)}1^{2}dx\right)  ^{1/2}\left(
\int_{X(t)}\rho^{2}dx\right)  ^{1/2}%
\end{equation}%
\begin{equation}
M=\int_{X(t)}\rho dx\leq\left\Vert X(t)\right\Vert ^{1/2}\left(  \int
_{X(t)}\rho^{2}dx\right)  ^{1/2}\leq(V_{\sup})^{1/2}\left(  \int_{X(t)}%
\rho^{2}dx\right)  ^{1/2}%
\end{equation}%
\begin{equation}
M^{2}\leq V_{\sup}\int_{X(t)}\rho^{2}dx
\end{equation}%
\begin{equation}
-\int_{X(t)}\rho^{2}dx\leq\frac{-M^{2}}{V_{\sup}}%
\end{equation}
for a bounded domain $\left\Vert X(t)\right\Vert \leq V_{\sup}<+\infty
$.\newline Thus, the inequality (\ref{Yuenineq}) becomes%
\begin{equation}
\frac{d}{dt}H\leq-\frac{1}{N}\int_{X(t)}\left(  \operatorname{div}u\right)
^{2}d\mu_{t}-\int_{X(t)}\rho^{2}dx+\Lambda M\leq-\frac{H^{2}}{MN}-\frac{M^{2}%
}{V_{\sup}}+\Lambda M
\end{equation}%
\begin{equation}
\frac{d}{dt}H\leq-\frac{H^{2}}{MN}-\frac{M^{2}}{V_{\sup}}+\Lambda M.
\label{ineq11new}%
\end{equation}
(1) If $\Lambda<M/V_{\sup}$, the Riccati inequality (\ref{ineq11new}) can be
estimated by%
\begin{equation}
\frac{d}{dt}H\leq-\frac{M^{2}}{V_{\sup}}+\Lambda M<0.
\end{equation}
Thus, there exists a finite time $T_{0}$, such that%
\begin{equation}
H(T_{0})<0.
\end{equation}
By applying the comparison property, we obtain%
\begin{equation}
\left\{
\begin{array}
[c]{c}%
\frac{d}{dt}H\leq-\frac{H^{2}}{MN}-\frac{M^{2}}{V_{\sup}}+\Lambda M\\
H(T_{0})<0.
\end{array}
\right.  \label{ineq11new2}%
\end{equation}
It is well known that the Riccati inequality (\ref{ineq11new2}) blows up at a
finite time $T$.\newline(2) If $\Lambda\geq M/V_{\sup}$ and $H(0)<-\sqrt
{-\frac{M^{3}N}{V_{\sup}}+\Lambda M^{2}N}$, it is also clear that the solution
of the Riccati inequality (\ref{ineq11new}) blows up at a finite time
$T$.\newline The proof is completed.
\end{proof}

\begin{remark}
For the one dimensional case, the condition $\Omega_{0ij}(x)=0$ in Theorem 1
is automatically satisfied.
\end{remark}

The corollary below is immediately shown in Theorem 1.

\begin{corollary}
For $\Lambda=0$, the non-trivial $C^{2}$ solutions with the bounded domain
$\left\Vert X(t)\right\Vert \leq V_{\sup},$ of the Euler-Poisson equations
(\ref{Euler-Poissonnew}) in $R^{N}$, and the initial condition%
\begin{equation}
\Omega_{0ij}(x)=\frac{1}{2}\left[  \partial_{i}u^{j}(0,x)-\partial_{j}%
u^{i}(0,x)\right]  =0,
\end{equation}
blow up at a finite time $T$.
\end{corollary}

\begin{remark}
By further requiring the bounded domain $\left\Vert X(t)\right\Vert \leq
V_{\sup}$ and
\begin{equation}
\Omega_{0ij}(x)=\frac{1}{2}\left[  \partial_{i}u^{j}(0,x)-\partial_{j}u_{{}%
}^{i}(0,x)\right]  =0,
\end{equation}
for the Euler-Poisson equations (\ref{Euler-Poissonnew}), the main achievement
of this spectral-dynamics-integration method is to remove the restriction on
$\operatorname{div}u(0,x_{0})$ with some point $x_{0}$, for the positive
background constant $\Lambda<M/V_{\sup}$ in Cheng and Tadmor's paper
\cite{CHENGT} for obtaining the blowup phenomenon.
\end{remark}

For the Euler-Poisson equations (\ref{Euler-Poissonnew}) with free boundaries,
it is possible to establish the existence of the solutions outside the bounded
domain $\left\Vert X(t)\right\Vert \leq V_{\sup}$ after the "blowup" time $T$
in Theorem \ref{thm:1}. Therefore, we have the following corollary:

\begin{corollary}
For the $N$-dimensional Euler-Poisson equations (\ref{Euler-Poissonnew}),
consider the non-trivial global $C^{2}$ solutions with $\rho(0,x)$ and
$u(0,x)$, which lie inside a bounded domain: $\left\Vert X(0)\right\Vert \leq
V_{0}$, where $\left\Vert \cdot\right\Vert $ denotes the volume and $V_{0}$ is
a positive constant. We define the weighted functional%
\begin{equation}
H(t)=\int_{X(t)}\operatorname{div}ud\mu_{t},
\end{equation}
with the positive measure $d\mu_{t}=\rho(t,x(t))dx(t)$. If the initial
condition
\begin{equation}
\Omega_{0ij}(x)=\frac{1}{2}\left[  \partial_{i}u^{j}(0,x)-\partial_{j}u_{{}%
}^{i}(0,x)\right]  =0
\end{equation}
and any one of the following conditions,\newline(1) $\Lambda<M/V_{0}$%
,\newline(2)\ $\Lambda\geq M/V_{0}$ and $H(0)<-\sqrt{-\frac{M^{3}N}{V_{0}%
}+\Lambda M^{2}N}$,\newline with the total mass $M=\int_{X(0)}\rho(0,x)dx>0$
are satisfied, $\left\Vert X(t)\right\Vert $ cannot be bounded by the constant
$V_{0}$ for all time $t$.
\end{corollary}

\section{Conclusions}

In this article, we study the life-span problem of self-gravitational fluids
with zero pressure (dust solutions) with a cosmological constant $\Lambda$ and
a bounded domain $X(t)$. We apply a new spectral-dynamics-integration method
to show that there are blowup phenomena if either the cosmological constant is
sufficiently small compared with other parameters of the pressureless
Euler-Poisson system (\ref{Euler-Poissonnew}), or if the weighted functional%
\begin{equation}
H(t)=\int_{X(t)}\rho\operatorname{div}udx
\end{equation}
is initially contracting sufficiently fast.

New functional techniques are expected to investigate the possibility of the
corresponding blowup phenomena for the Euler-Poisson equations with the
pressure term:
\begin{equation}
\left\{
\begin{array}
[c]{rl}%
{\normalsize \rho}_{t}{\normalsize +\nabla\cdot(\rho u)} & {\normalsize =}%
{\normalsize 0}\\
\rho\lbrack u_{t}+(u\cdot\nabla)u]+K\nabla\rho^{\gamma} & {\normalsize =-}%
{\normalsize \rho\nabla\Phi}\\
{\normalsize \Delta\Phi(t,x)} & {\normalsize =}{\normalsize \rho-\Lambda,}%
\end{array}
\right.
\end{equation}
with constants $K>0$ and $\gamma\geq1$.

\section{Acknowledgement}

The author thanks the reviewers for their helpful comments for improving the
quality of this article. This work is partially supported by the Dean's
Research Fund FLASS/ECR-9 of the Hong Kong Institute of Education.

\end{document}